\documentclass{llncs} % instead of \documentclass[11pt]{article}
\newtheorem{defn}{Definition} %[section]
\newtheorem{df}{Definition}%[defn]

\newtheorem{lem}[defn]{Lemma}
\newtheorem{cor}{Corollary}
\newtheorem {rem}{Remark}
\newtheorem{con}{Conjecture}

\title{Matrix approach to synchronizing automata}

\date{}
\author{A.N. Trahtman\thanks{Email: avraham.trakhtman@gmail.com}
\institute{22.10.2019}
}
\begin{document}

\maketitle
  \begin{abstract}
  A word $w$ of letters on edges of underlying graph $\Gamma$ of deterministic finite automaton (DFA) is called synchronizing if
$w$ sends all states of the automaton to a unique state.

J. \v{C}erny discovered in 1964 a sequence of $n$-state complete
DFA possessing a minimal synchronizing word of length $(n-1)^2$.
The hypothesis, well known today as \v{C}erny conjecture, claims
that $(n-1)^2$ is a precise upper bound on the length of such a
word over alphabet $\Sigma$ of letters on edges of
$\Gamma$ for every complete $n$-state DFA. The hypothesis was
formulated distinctly in 1966 by Starke.

A special classes of matrices induced by words in the alphabet of labels on edges of the underlying  graph of DFA
are used for the study of  synchronizing automata.

\end{abstract}

$\bf Keywords $ deterministic finite automata, synchronizing
word, \v{C}erny conjecture.

\section*{Introduction}

 The problem of synchronization of DFA is a natural one and various aspects of this problem have been touched in the literature.
The connections with the early coding theory and first efforts to
estimate the length of synchronizing word look in the works
\cite{La}, \cite{Li}. Prehistory of the topic, the emergence of
the term, different problems of synchronization one can find in
surveys  \cite{Ju}, \cite{KV}, \cite{Vo1}, \cite{Vo}.

Synchronization makes the behavior of an automaton resistant
against input errors since, after detection of an error, a
synchronizing word can reset the automaton back to its original
state, as if no error had occurred. The synchronizing word limits
the propagation of errors for a prefix code. Deterministic finite
automaton is a tool that helps to recognise language in a set of
DNA strings.

A problem with a long story is the estimation of the minimal length of synchronizing word.

 J. \v{C}erny in 1964 \cite{Ce} found the sequence of $n$-state complete DFA with shortest synchronizing word of length $(n-1)^2$
 for an alphabet of size two.
  The hypothesis, well known today as the \v{C}erny's conjecture, claims that this lower bound on the length of the synchronizing word
  of aforementioned automaton is also the upper bound for the shortest synchronizing word of any $n$-state complete DFA:

\begin{con}
The deterministic complete $n$-state synchronizing automaton over alphabet $\Sigma$ has synchronizing word in $\Sigma$
of length at most $(n-1)^2$ \cite{Sta} (Starke, 1966).
  \end{con}

The problem can be reduced to automata with a strongly connected graph \cite{Ce}.

This famous conjecture is true for a lot of automata, but in general the problem still remains open although several hundreds
of articles consider this problem from different points of view \cite{TB}.
Moreover, two conferences "Workshop on Synchronizing  Automata"
(Turku, 2004) and  "Around the \v{C}erny conjecture"
(Wroclaw, 2008) were dedicated to this longstanding conjecture.
The problem is discussed on many sites on the Internet.

Together with the Road Coloring problem \cite{MS}, \cite{Fi}, \cite{PS}, this simple-looking conjecture was arguably
the most longstanding and famous open combinatorial problems in the theory of finite automata \cite{KV}, \cite{Pin}, \cite{Sta}, \cite{St}, \cite{Vo}.

The road coloring problem to find a labelling of the edges that turns the graph into a deterministic finite automaton
possessing a synchronizing word was stated in 1970 \cite{AW} and solved in 2008 \cite{TP}, \cite{TI}.

Examples of automata such that the length of the shortest
synchronizing word is greater than $(n-1)^2$ are unknown for
today. Moreover, the examples of automata  with shortest
synchronizing word of length $(n-1)^2$ are infrequent. After the
sequence of \v{C}erny and the example of \v{C}erny, Piricka and
Rosenauerova \cite{CPR} of 1971 for $|\Sigma|=2$, the next such
examples were found by Kari \cite {Ka} in 2001 for $n=6$ and
$|\Sigma|=2$ and by Roman \cite {Ro} for $n=5$ and $|\Sigma|=3$ in
2004.

The package TESTAS \cite {TS}, \cite {Tt} studied all automata with strongly connected underlying graph of size $n \le 11$ for
 $|\Sigma|=2$, of size $n \le 8$ for $|\Sigma| \le 3$ and of size $n \le 7$ for $|\Sigma| \le 4$ and found five new examples of
 DFA with shortest synchronizing word of length $(n-1)^2$ for $n\leq 4$.

 Don and Zantema present in \cite {DZ} an ingenious method of designing several new automata from existing examples of size three and four
 and proved that for $n\geq 5$ the method does not work.
So there are up to isomorphism exactly 15 DFA for $n=3$ and 12 DFA for $n=4$ with shortest synchronizing word of length $(n-1)^2$.
The authors of \cite {DZ} support the hypothesis from \cite{TS} that all automata with shortest synchronizing word of length $(n-1)^2$ are known,
of course, with essential correction found by themselves for
$n=3,4$.

There are several reasons \cite{AGV}, \cite{BBP}, \cite{CA},  \cite {DZ}, \cite{TS} to believe that the length of the shortest synchronizing word
for remaining automata  with $n>4$ (except the sequence of
\v{C}erny and two examples for $n=5, 6$) is essentially less and the gap grows with $n$.
For several classes of automata, one can find some estimations on the length in \cite{AGV}, \cite{GJV}, \cite{CR}, \cite{Kr}, \cite{KKS}, \cite{Ta}.

Initially found upper bound for the minimal length of
synchronizing word was big and
has been consistently improved over the years by different authors.
The upper bound found by Frankl in 1982 \cite{Fr} is equal to
$(n^3-n)/6$.
The result was reformulated in terms of synchronization in \cite{Pin} and repeated independently in \cite{KRS}.

The mentioned results for $(n^3-n)/6$ successfully use the matrix approach and the dimension of the arising spaces.
See also, for instance, \cite{BBP}, \cite{AS}, \cite{BS}, \cite{KV}, \cite{Jun12}, \cite{GJ16}, \cite{GGJ18} for this approach.

Nevertheless, the cubic estimation of the bound exists since 1982.

The considered deterministic automaton $A$ can be presented by a
complete underlying graph with edges labelled by letters of an
alphabet.

We consider a special class of matrices $M_u$ of mapping induced by words $u$ in the alphabet of letters on edges
of the underlying graph.
$M_u$ has one unit in every row and rest zeros.
The matrix of synchronizing word has units only in one column.

Our proof used some lemmas from \cite{TM}. For a complete picture of the proof, these lemmas after some modification are included in the proposed work.

Help Lemmas \ref{l2} and \ref{po} state that
the size of the set $R(u)$ of nonzero columns of the matrix $M_u$ is equal to the rank of $M_u$, $R(bu)\subseteq R(u)$ and
$|R(ub)| \leq |R(u)|$ for every word $b$.

Lemma \ref {v3} estimates the dimension of the space generated by matrices of words:
The set of all $n\times k$-matrices of words
for $k<n$ has at most $n(k-1)+1$ linear independent matrices.

In particular, the set of $n\times (n-1)$-matrices of words
has at most $(n-2)^2$ linear independent matrices. The famous value from the \v{C}erny hypothesis appears here.

Lemma \ref {lam} studied the nontrivial linear combination of matrices of words:
\begin{equation}
M_u =\sum_{i=1}^k\lambda_i M_{u_i}\to \sum^k_{i=1}\lambda_i =1.
\quad \sum^k_{i=1}\lambda_iM_{u_i}=0 \to \sum^k_{i=1}\lambda_i=0.
  \label{lmi}
\end{equation}

Lemma \ref {l3} notes distributivity by multiplication matrix from left on linear combination of matrices of word.

We study the rational series $(S,u)$ for matrix $M_u$ (see \cite{BR}), \cite{Be}.
This approach for synchronizing  automata supposed first by B{\'e}al \cite{Be} proved to be fruitful \cite{BBP}, \cite{CA}, \cite{Co}.

Lemma \ref{v4} and its Corollary \ref{c4} establish some algebraic properties of rational series of matrices of words,
for instance:

the matrices $M_{u}$ with constant $(S,u)=i$ generate a space
$V$ with $(S,t)=i$ for every nonzero matrix $M_t \in V$.

We consider the equations $M_uL_x=M_s$ (\ref{ux}) for synchronizing word $s=ux$ and solutions $L_x$
(Definition \ref{dL}) in Lemma \ref{l5}.

A connection between the set of nonzero columns of matrix of word, subsets of states of automaton
and of solutions $L_x$  of (\ref{ux}) is revealed in Remarks
\ref{r4}, \ref{r7}.

Lemmas \ref{l9}, \ref{l7} consider pseudoinverse matrices
(\ref{inv}) and their connection with equation (\ref{ux}).

The ideas of the approach are illustrated on examples of automata with a maximal length of synchronizing word from \cite{Ka}, \cite{Ce}, \cite{Ro}.
A sequence of words $u$ of growing length together with corresponding $n$-vector of subset of states obtained
by mapping of $u$ presents column $q$ of solutions from (\ref{ux}).
Some connection between the sequence of linear independent solutions ($n$-vector of subset of states) and subwords of
the minimal synchronizing word is easy to detect.

 \section*{Preliminaries}

We consider a complete $n$-state DFA with
 strongly connected underlying graph $\Gamma$ and transition semigroup $S$
 over a fixed finite alphabet $\Sigma$ of labels on edges of  $\Gamma$ of an automaton $A$.
The trivial cases $n \leq 2$, $|\Sigma|=1$ and $|A \sigma|=1$ for $\sigma \in\Sigma$ are excluded.

The restriction on strongly connected graphs is based on \cite{Ce}.
The states of the automaton $A$ are considered also as vertices of the graph $\Gamma$.

If there exists a path in an automaton from the state $\bf p$ to
the state $\bf q$ and the edges of the path are consecutively
labelled by $\sigma_1, ..., \sigma_k$, then for
$s=\sigma_1...\sigma_k \in \Sigma^+$ let us write ${\bf q}={\bf
p}s$.

Let $Px$ be the set of states ${\bf q}={\bf p}x$ for all ${\bf p}$
from the subset $P$ of states and $x \in \Sigma^+$.
Let $Ax$ denote the set $Px$ for the set $P$ of all states of the automaton.

 A word $s \in \Sigma^+ $ is called a {\it synchronizing (reset, magic, recurrent, homing, directable)} word
of an automaton $A$ with underlying graph $\Gamma$ if $|As|=1$.
The word $s$ below denotes minimal synchronizing word such that for a state $\bf q$ $As=\bf q$.

The states of the automaton are enumerated with number one for the fixed state $\bf q$.

 An automaton (and its underlying graph) possessing a synchronizing word is called {\it synchronizing}.

Let us consider a linear space generated by
$n\times n$-matrices $M$ with one unit in any row of the matrix and zeros everywhere else.

We connect a mapping of the set of states of the automaton
made by a word $u$ of $n\times n$-matrix $M_u$ such that for an element $m_{i,j} \in M_u$ takes place

\centerline{$m_{i,j}=1$ if the word $u$ maps ${\bf q}_i$ on ${\bf q}_j$ and $0$ otherwise.}

Any mapping of the set of states of the automaton  $A$ can be presented by some word $u$
and by a corresponding matrix $M_u$.
For instance,

 \centerline{$M_u=\left(
\begin{array}{ccccccc}
  0 & 0 & 1 & . & . & . &  0 \\
  1 & 0 & 0 & . & . & . &  0 \\
  0 & 0 & 0 & . & . & . &  1 \\
  . & . & . & . & . & . &  . \\
  0 & 1 & 0 & . & . & . &  0 \\
  1 & 0 & 0 & . & . & . &  0 \\
\end{array}\right)
$}

 Let us call the matrix $M_u$ of the mapping induced by the word $u$, for brevity, the matrix of word $u$, and vice versa,
$u$ is the word of matrix $M_u$.

Every square matrix with one unit in every row and zeros in remaining cells will be also considered as a matrix of word.

$M_uM_v=M_{uv}$ \cite{Be}.

The set of nonzero columns of $M_u$ (set of second indexes of its elements) of $M_u$ is denoted as $R(u)$.

The word $u$ of the matrix $M_u$ is called {\it irreducible}
if for every word $v$ obtained by removing some subword of $u$ $M_u \neq M_v$.

The minimal synchronizing word and all its subwords are irreducible.

The right word $x$ of synchronizing word $ux$ let us call right synchronizing continuation of $u$.

Zero matrix is a matrix of empty word.

The subset of states $Au$ of the set of all states of $A$ is denoted $c_u$ with number of states $|c_u|$.
In $n$-vector $c_u$ the coordinate $j$ has unit if
the state $j \in c_u$ and zero in opposite case.

For linear algebra terminology and definitions, see \cite{Ln},
\cite{Ma}.

\section{Mappings induced by a word and subword}

\begin{rem} \label{r1}
For every cell of $n\times n$-matrix of words in strongly connected automaton there is a matrix with unit in the cell.

Every unit in the product $M_uM_a$ is the product of two units, first unit from a row of $M_u$
and second unit from nonzero column of $M_a$.

The unit in the cell $(i,j)$ of the matrix of letter denotes the edge from the state $i$ to the state $j$,
for matrix of word, such unit denotes the path from $i$ to $j$.

The set $R(u)$ of nonzero columns of matrix $M_u$ corresponds the set of states $c_u$ of the automaton.
\end{rem}

\begin{lem}  \label {po}
For the set of states of deterministic finite automaton and any words $u$ and $w$ $Auw \subseteq Aw$.

For every word $w$, $R(u)\subset R(v)$ implies
$R(uw)\subseteq R(vw)$.

For the state ${\bf p} \not\in Aw$ takes place
 ${\bf p} \not\in Auw$.
 Nonzero columns of $M_{uw}$ have units also in $M_w$.

The number of nonzero columns $|R(b)|$ is equal to the rank of $M_b$.

\end{lem}

\begin{proof}
The properties of $Au \subseteq A$, $M_w$ and $M_{uw}$ follow from the definition of the matrix of word.

The set of nonzero columns of matrix defines a set of states.
The mapping by word $w$ of a set of states [columns $R(v)$] induces a mapping of its subset [columns $R(u)$].

For any word $u$ and the zero column of $M_w$
the corresponding column of $M_{uw}$ also consist of zeros.
Hence nonzero columns of $M_{uw}$ have units also in $M_w$.

The matrix $M_b$ has submatrix with only one unit in every row and every nonzero column with nonzero determinant.
Therefore $|R(b)|$ is equal to the rank of $M_b$.

\end{proof}

\begin{cor}  \label{c1}
The matrix $M_s$ of word $s$ is synchronizing if and only if $M_s$ has zeros in all columns except one and units in the residuary column.

All matrices of right subwords of $s$ also have at least one unit in this column.
\end{cor}

\begin{rem} \label{r4}
The rows of the matrix $M_aM_u$ are obtained by permutation of rows of the matrix $M_u$.
Some of these rows may disappear and replaced by another rows of $M_u$.

The invertible matrix $M_a$ does not change the number of units of every column of $M_u$ in its image of
the product $M_aM_u$.

The columns of the matrix $M_uM_a$ are obtained by permutation of columns $M_u$.
Some columns can be merged with $|R(ua)|<|R(u)|$.

\end{rem}

\begin{lem} \label{l2}
For every words $a$ and $u$

\centerline{$|R(ua)| \leq |R(u)|$ and }

\centerline{$R(au) \subseteq R(u)$.}
The matrix $M_a$ with $m$ units in column $r$ replicates row $r$ of $M_u$ $m$ times in $M_aM_u$.

For invertible matrix $M_a$  $R(au)=R(u)$ and  $|R(ua)|=|R(u)|$.

\end{lem}

\begin{proof}
The matrix $M_a$ in the product $M_uM_a$ shifts column of $M_u$
to columns of $M_uM_a$ without changing the column itself
(Remark \ref{r4}).

$M_a$ also can merge columns of $M_u$.
In view of possible merged columns, $|R(ua)|\leq |R(u)|$.

The zero columns $j$ of $M_a$ changes the row $j$ of $M_u$ in the product $M_aM_u$.

Let $M_a$ have $m$ units in column $r$.
These units and unit in row $r$ of $M_u$ create $m$ units in the product $M_aM_u$ in different rows of common column.
Therefore the matrix $M_a$ replicates the row $r$ of $M_u$ $m$ times in $M_aM_u$.

So some rows of $M_u$ can be replaced in $M_aM_u$ by row $r$ and therefore some rows from $M_u$ may disappear (Remark \ref{r4}).

Hence $R(au) \subseteq R(u)$ (See also Lemma \ref{po}).

For invertible matrix $M_a$ in view of existence $M_a^{-1}$ we have $|R(ua)|= |R(u)|$ and $R(au)= R(u)$.

\end{proof}

\section{The set of linear independent matrices of words}

\begin{rem} \label{r2}
The space generated by matrices of words has zero matrix of empty word.
\end{rem}

\begin{lem}  \label {v3}
 The set $V$ of all $n\times k$-matrices of words
(or $n\times n$-matrices with zeros in fixed $n-k$ columns for $k<n$) has $n(k-1)+1$ linear independent matrices.
 \end{lem}
\begin{proof}
Let us consider distinct $n\times k$-matrices of word with at most only one nonzero cell outside the last nonzero column $k$.

Let us begin from the matrices $V_{i,j}$ with unit in $(i,j)$ cell ($j<k$) and units in ($m,k$) cells for all $m$ except $i$.
The remaining cells contain zeros.
So we have $n-1$ units in the $k$-th column and only one unit in remaining $k-1$ columns of the matrix $V_{i,j}$.
Let the matrix $K$ have units in the $k$-th column and zeros in the other columns.
There are $n(k-1)$ matrices $V_{i,j}$. Together with $K$ they belong to the set $V$.
So we have $n(k-1)+1$ matrices. For instance,

\begin{picture}(0,40)
\end{picture}
$V_{1,1}=\left(
\begin{array}{cccccccc}
  1 & 0 & 0 & . & . & 0  \\
  0 & 0 & 0 & . & . & 1  \\
  0 & 0 & 0 & . & . & 1  \\
  . & . & . & . & . & .  \\
  0 & 0 & 0 & . & . & 1  \\
  0 & 0 & 0 & . & . & 1  \\
\end{array}
\right)$
\begin{picture}(4,40)
\end{picture}
$V_{3,2}=\left(
\begin{array}{cccccccc}
  0 & 0 & 0 & . & . & 1  \\
  0 & 0 & 0 & . & . & 1  \\
  0 & 1 & 0 & . & . & 0  \\
  . & . & . & . & . & .  \\
  0 & 0 & 0 & . & . & 1  \\
  0 & 0 & 0 & . & . & 1  \\
\end{array}
\right)$
\begin{picture}(4,40)
\end{picture}
$K=\left(
\begin{array}{cccccccc}
  0 & 0 & 0 & . & . & 1 \\
  0 & 0 & 0 & . & . & 1 \\
  0 & 0 & 0 & . & . & 1 \\
  . & . & . & . & . & . \\
  0 & 0 & 0 & . & . & 1 \\
  0 & 0 & 0 & . & . & 1 \\
\end{array}
\right)$

 The first step is to prove that the matrices $V_{i,j}$ and $K$ generate the space with the set $V$.
For arbitrary matrix $T$ of word from $V$ for every $t_{i,j} \neq 0$ and $j<k$,
let us consider the matrices $V_{i,j}$ with unit in the cell $(i,j)$ and the sum of them $\sum V_{i,j}=Z$.

The first $k-1$ columns of $T$ and $Z$ coincide.
   Hence in the first $k-1$ columns of the matrix $Z$ there is at most only one unit in any row.
 Therefore in the cell of $k$-th column of $Z$ one can find at most two values which differ by unit, say $m$ or $m-1$.
The value of $m$ appears if there are only zeros
in other cells of the considered row. Therefore $\sum V_{i,j} - (m-1)K=T$.

Thus, every matrix $T$ from the set $V$ is a span of
above-mentioned $(k-1)n +1$ matrices from $V$.
It remains now to prove that the set of matrices $V_{i,j}$ and $K$ is a set of linear independent matrices.

If one excludes a certain matrix $V_{i,j}$ from the set of these matrices, then it is impossible
to obtain a nonzero value in the cell $(i,j)$ and therefore to obtain the matrix $V_{i,j}$.
So the set of matrices $V_{i,j}$ is linear independent.
Every non-trivial linear combination of the matrices $V_{i,j}$ equal to a matrix of word has at least one nonzero element
in the first $k-1$ columns.
Therefore, the matrix $K$ could not be obtained as a linear combination of the matrices $V_{i,j}$.
Consequently the set of matrices $V_{i,j}$ and $K$ forms a basis of the set $V$.
\end{proof}

\begin{cor}  \label {c2}
The set of all $n \times(n-1)$-matrices of words
(or $n\times n$-matrices with zeros in a fixed column)
has $(n-1)^2$ linear independent matrices.
 \end{cor}
Proof. For $k=n-1$ it follows from $n(n-1-1)+1= (n-1)^2$.

\begin{cor}  \label {cp}
Suppose the vertex ${\bf  p} \not\in A \alpha$ and let words $u$ of matrices $M_u$ have the last letter $\alpha$.

Then there are at most $(n-1)^2$ linear independent matrices $M_u$.
 \end{cor}
Proof. All matrices $M_u$ have common zero column ${\bf p}$ by Lemma \ref{po}.
So we have $n\times n$-matrices with zeros in a fixed column
and due to Corollary \ref{c2}
there are at most $(n-1)^2$ linear independent matrices $M_u$.

 \begin{cor}  \label {c3}
There are at most $n(n-1)+1$ linear independent matrices of words in the set of $n\times n$-matrices.
 \end{cor}

\begin{cor}  \label {c3a}
There are at most $n+1$ linear independent matrices of words in the set of matrices with 2 nonzero columns.
 \end{cor}

\begin{lem}\label{lam}
 Suppose that for nonzero matrices $M_u$ of word $u$ and $M_{u_i}$ of words $u_i$
\begin{equation}
M_u =\sum_{i=1}^k\lambda_i M_{u_i}. \label{lm}
\end{equation}
Then the sum $\sum^k_{i=1}\lambda_i =1$ and the sum $S_j$ of values in every row $j$ of the sum in (\ref{lm}) also is equal to one.

If $\sum^k_{i=1}\lambda_iM_{u_i}=0$  then $\sum_{i=1}^k \lambda_i=0$ and $S_j=0$ for every $j$ with $M_u=0$.

If the sum $\sum^k_{i=1}\lambda_i$ in every row is not unit
[zero] then $\sum_{i=1}^k\lambda_i M_{u_i}$
is not a matrix of word.
\end{lem}
\begin{proof}
The nonzero matrices $M_u$ and $M_{u_i}$ have $n$ cells with unit in the cell.
Therefore, the sum of values in all cells of the matrix $\lambda_i M_{u_i}$ is $n \lambda_i$.

For nonzero $M_u$ the sum is $n$. So one has in view of
$M_u =\sum_{i=1}^k\lambda_i M_{u_i}$

\centerline {$n=n\sum_{i=1}^k \lambda_i$, whence $1 =\sum_{i=1}^k \lambda_i$.}
Let us consider the row $j$ of matrix $M_i$ in (\ref{lm}) and let  $1_i$ be unit in the row $j$.
The sum of values in a row of the sum (\ref{lm}) is equal to unit in the row of $M_u$.
So $1 =\sum_{i=1}^k \lambda_i1_i=\sum_{i=1}^k \lambda_i$.

$\sum_{i=1}^k\lambda_i M_{u_i}=0$ implies $S_j=\sum_{i=1}^k \lambda_i1_i=\sum_{i=1}^k \lambda_i=0$ for every row $j$.

If the matrix $M=\sum_{i=1}^k\lambda_i M_{u_i}$ is a matrix
of word or zero matrix then $\sum^k_{i=1}\lambda_i \in \{0, 1\}$.
If $\sum^k_{i=1}\lambda_i\not\in \{0, 1\}$
or the sum is not the same in every row then we have opposite case.
\end{proof}

\begin{lem} \label{l3} {\it Distributivity from left.}

For every words $b$ and $x_i$

\centerline{$M_b\sum\tau_iM_{x_i}=\sum\tau_iM_bM_{x_i}$.}
\end{lem}

\begin{proof}
The matrix $M_b$ shifts rows of every $M_{x_i}$ and of the sum of them in the same way according to Remark \ref{r4}.
$M_b$ removes common row of them and replace also by common row
(Remark \ref{r4}).

Therefore the matrices $M_bM_{x_i}$ has the origin rows of
$M_{x_i}$, maybe in another order,
and the rows of the sum $\sum\tau_iM_bM_{x_i}$ repeat rows of  $\sum\tau_iM_{x_i}$ also in the same order.

\end{proof}

Note that this is not always true on the right.

\section{Rational series}
The section follows ideas and definitions from \cite{BR} and \cite{Be}.
We recall that a formal power series with coefficients in a field $K$ and variables in $\Sigma$ is
a mapping of the free monoid $\Sigma^*$ into $K$ \cite{BR}, \cite{CA}.

We consider an $n$-state automaton $A$. Let $P$ denote the subset
of states of the automaton with the characteristic column vector
$P^t$ of $P$ of length $n$ having units in coordinates corresponding to the states of $P$ and zeros everywhere else.
Let $C$ be a row of units of length $n$.
 Following \cite{Be}, we denote by $S$ the {\it rational series} depending on the set $P$ defined by:
\begin{equation}
(S,u) = C M_uP^t-C P^t= C(M_u-E)P^t. \label{ser}
\end{equation}

\begin{rem} \label{r3}
Let $S$ be a rational series depending on the set $P$.

If the cell $i$ in $P^t$ has zero
then $(S,u)$ does not depend on column $i$ of $M_u$.
If this cell $i$ has unit then the column $i$ of $M_u$ with $k$ units from (\ref{ser}) added to $(S,u)$ the value of $k-1$.

For $k$ units in the column $q$ of $n\times n$-matrix $M_u$ and $P=\{\bf q\}$ $(S,u)=k-1$.

For $P$ nonzero columns of $n\times n$-matrix $M_u$ $(S,u)=n-|P|$.

\end{rem}

\begin{lem} \label{v4}
Let $S$ be a rational series depending on the
set $P$ of an automaton $A$.
Let $M_u=\sum_{j=1}^k\lambda_j M_{u_j}$.
Then $(S,u) =\sum_{j=1}^k\lambda_j (S,u_j)$.

If $(S,u_j)=i$ for every $j$ then also $(S,u)=i$.
\end{lem}
\begin{proof}
One has in view of (\ref{ser})

\centerline {$(S,u)= C(\sum^k_{j=1}\lambda_jM_{u_j}-E)P^t$}
where $C$ is a row of units and $P^t$ is a characteristic
column of units and zeros.

Due to Lemma \ref{lam}

$\sum^k_{j=1}\lambda_jM_{u_j}-E=\sum^k_{j=1}\lambda_jM_{u_j}-\sum^k_{j=1}\lambda_j E =
\sum^k_{j=1} \lambda_j(M_{u_j}-E)$.
So
$(S,u)=C(M_u-E)P^t = C(\sum^k_{j=1}\lambda_j M_{u_j}-E)P^t =
C(\sum^k_{j=1}\lambda_j (M_{u_j}-E))P^t=
 \sum^k_{j=1}\lambda_jC(M_{u_j}-E)P^t=
\sum^k_{j=1}\lambda_j(S, u_j)$.

Thus, $(S,u) =\sum_{j=1}^k\lambda_j (S,u_j)$.

If $\forall j$ $(S,u_j)=i$, then
$(S,u)= \sum^k_{j=1}\lambda_j i=i\sum^k_{j=1}\lambda_j=i$ by Lemma \ref{lam}.
\end{proof}

From Lemma \ref{v4} follows
\begin{cor} \label{c4}
Let $S$ be a rational series depending on the
set $P$ of an automaton $A$.

The matrices $M_{u}$ with constant $(S,u)=i$ generate a space
$V$ such that for every nonzero matrix $M_t \in V$ of word $t$ $(S,t)=i$.
\end{cor}

\begin{cor}\label{c5}
Let $S$ be a rational series depending on the set $P$ of size one of $n$-state automaton.

Then the set $V$ of matrices $M_{u}$ with two fixed nonzero columns and fixed nonnegative $(S,u) <n-1$ has at most $n$ linear independent matrices.
\end{cor}
\begin{proof}

By lemma \ref{v3} for $k=2$ there are at most $n+1$ linear independent matrices.
There is a matrix $M_w$ in a space for $k=2$ with one nonzero column and $(S,w)\neq (S,u)$.
Therefore fixed $(S,u)<n-1$ excludes the matrix $M_w$ from space generated by $V$.
\end{proof}

\section{Matrix $L_x$ of word $x$ with common column $q$ of
$M_x$.}
Let $S_q$ below be a rational series depending on the set
$P =\{\bf q\}$ of size one
with number one for the column $q$.

By Remark \ref{r3} the matrix $M_x$ has $(S_q,x)+1$ units in the column $q$.

\begin{df} \label{sq}

If the set of cells with units in
the column $\bf q$ of the matrix $M_v$ is a subset of the
analogous set of the matrix $M_u$ then we write

\centerline{$M_v \sqsubseteq_q M_u$}

\end{df}

\begin{df} \label{dL}
The matrix $L_x$ has column $q$ of $M_x$. For $(S_q,x)=n-i$ with $0<i\leq n$, $L_x$ has $n-i+1$ units in the column one of
the state $\bf q$
and remaining $i-1$ units in next nearest columns, as usual one unit in the column.
\end{df}

$M_x=\left(
\begin{array}{cccccccc}
  0 & 1 & 0 & 0 & 0 \\
  1 & 0 & 0 & 0 & 0 \\
  1 & 0 & 0 & 0 & 0 \\
  0 & 0 & 0 & 0 & 1 \\
  0 & 0 & 1 & 0 & 0 \\
\end{array}
\right)$
$L_x=\left(
\begin{array}{cccccccc}
  0 & 0 & 1 & 0 & 0 \\
  1 & 0 & 0 & 0 & 0 \\
  1 & 0 & 0 & 0 & 0 \\
  0 & 0 & 0 & 1 & 0 \\
  0 & 1 & 0 & 0 & 0 \\
\end{array}
\right) (S_q,x)=1$

The set of matrices $L_x$ with nonzero column  $q$ is a subset of the set of matrices of word,
having in every row one unit and the rest of zeros.

For $(S_q,t)=n-2$ and suitable numeration of columns some
$L_t=M_t$.

Many results concerning matrices $M_x$ are valid for $L_x$
because they use only matrix properties and existence of
one unit in a row.
For instance, Lemma \ref{lam}, Corollary \ref{c3a} of Lemma
\ref{v3}, Corollary \ref{c4} of Lemma \ref{v4} directly related to the set of matrices $L_x$.

\subsection{The equation with unknown $L_x$}

The solution $L_x$ of the equation
\begin{equation}
M_uL_x=M_s \label{ux}
\end{equation}
for synchronizing matrix $M_s$ and arbitrary $M_u$ must have units in the column of the state $\bf q$.

In general, there are several solutions $L_x$ of synchronizing continuations $x$ of the word $u$.

\begin{lem}  \label{l5}
Every equation $M_uL_x=M_s$ (\ref{ux}) has a solutions $L_x$ with $(S_q,x) \geq 0$.

 $|R(u)|-1=(S_q,x)$ for $L_x$ with minimal $(S_q,x)$ (a minimal solution),
every matrix $L_y$ satisfies the equation (\ref{ux}) iff
$L_x\sqsubseteq_q L_y$.

Distinct solutions $L_x$ can differ in rows corresponding zero columns of $M_u$.

The rank $|R(x)|\leq n-(S_q,x)$.

There exists one-to-one correspondence between nonzero columns of $M_u$, units in the column $q$ of minimal solution $L_x$
and the set of states $c_u=Au$ of automaton $A$.

\end{lem}

\begin{proof}
The matrix $M_s$ of rank one has column of units of the state $\bf q$.
For every nonzero column $j$ of $M_u$ with elements $u_{i,j}=1$ and $s_{i,q}=1$ in the matrix $M_s$ let the cell $(j,q)$
have unit in the matrix $L_x$.
So the unit in the column $q$ of matrix $M_s$ is a product of
every unit from the column $j$ of $M_u$
and unit in the cell $j$ of column $q$ of $L_x$.

Therefore by Remark \ref{r3} for rational series $S_q$ the minimal solution $L_x$ has in the column $q$
$(S_q,x)+1$ units, whence $(S_q,x)=|R(u)|-1$.

The set $R(u)$ of nonzero columns of $M_u$ corresponds the set of cells of the column $q$ with unit of minimal $L_x$.

So to the column $q$ of every solution belong at least $(S_q,x)+1$ units.
The units outside column $q$ of the solution $L_x$ belong to the next columns, one unit in a row.
Units in rows corresponding zero columns of $M_u$ do not imply on result in (\ref{ux}) (Remark \ref{r1})
and therefore can be placed arbitrarily.
The remaining cells obtain zero.

Lastly every solution $L_x$ of (\ref{ux}) has one unit with rest of zeros in every row and is a matrix of word.

Zeros in the cells of column $q$ of minimal $L_x$ correspond zero columns of $M_u$.
Therefore for the matrix $L_y$ such that $L_x \sqsubseteq_q L_y$
we have $M_uL_y=M_s$.
Every solution $L_y$ must have units in cells of column $q$ that correspond $|R(u)|=(S_q,x)+1$ nonzero columns of $M_u$.
Therefore $R(x)$ in (\ref{ux}) has at most $n-(S_q,x)-1$ nonzero columns besides $q$,
whence the rank $|R(x)|\leq n-(S_q,x)$.

Thus, the equality $M_uL_x=M_uL_y=M_s$ is
equivalent to $L_x \sqsubseteq_q L_y$ for the minimal $L_x$.

The matrix $M_u$ with set $R(u)$ of nonzero columns maps
the automaton on the set $c_u$ of states and on the set of units in the column $q$ of minimal $L_x$.

\end{proof}

\begin{cor} \label{c7}

For minimal solution $L_x$ of the equation
$M_uL_x=M_s$ and minimal solution $L_y$ of the equation
$M_{ut}L_y=M_s$ one has
$(S_q,y)\leq (S_q,x)$ in view of
$|R(ut)|\leq |R(u)|$ (Lemma \ref{l2}).

\end{cor}

Lemma \ref{l5} explains the following

\begin{rem} \label{r7}

Every permutation and shift of nonzero columns $M_u$
induces corresponding  permutation of the set of units in
the column $q$ of minimal solution $L_x$ of (\ref{ux}), and vice versa.
\end{rem}

\section{Right pseudoinverse matrices}

\begin{df} \label{inv}

Let us call the matrix $M_{a^-}$ of word
{\sf right pseudoinverse} matrix of the matrix $M_a$
of a word $a$ if for precisely one element $a_{i,j}=1$ of
every nonzero column $j$ of $M_a$ the cell $(j,i)$ of
$M_{a^-}$ has unit.

If in $M_{a^-}$ still are zero rows then unit is
added arbitrary in every such row of the matrix $M_{a^-}$
of word.

Let $E_a$ denote $M_aM_{a^-}$.
\end{df}

For instance,
\\
\\
\noindent$M_a=\left(
\begin{array}{cccccccc}
  0 & 1 & 0 & 0 & 0 \\
  0 & 1 & 0 & 0 & 0 \\
  0 & 0 & 0 & 0 & 1 \\
  0 & 0 & 1 & 0 & 0 \\
  0 & 0 & 1 & 0 & 0 \\
\end{array}
\right)$
$M_{a^-}=\left(
\begin{array}{cccccccc}
  0 & 1 & 0 & 0 & 0 \\
  \d{1} & 0 & 0 & 0 & 0 \\
  0 & 0 & 0 & 0 & \d{1} \\
  0 & 0 & 0 & 1 & 0 \\
  0 & 0 & \d{1} & 0 & 0 \\
\end{array}
\right)$
$M_{a^-}=\left(
\begin{array}{cccccccc}
  1 & 0 & 0 & 0 & 0 \\
  0 & \d{1} & 0 & 0 & 0 \\
  0 & 0 & 0 & \d{1} & 0 \\
  0 & 0 & 0 & 1 & 0 \\
  0 & 0 & \d{1} & 0 & 0 \\
\end{array}
\right)$

\begin{rem} \label{ri}

For invertible matrix $M_a$ with $|R(a)|=n$ we have
$M_{a^-}=M_a^{-1}$, for singular $M_a$ there are some
generalized inverse matrices, including invertible.

Pseudoinverse matrices can be considered as matrices of word in the alphabet $\Gamma^-$.

\end{rem}

\begin{lem} \label{l9}

The product $M_bM_{b^-}=E_b$ does not depend on any arbitrary placing of units in $M_{b^-}$ in empty rows.
The rank of $E_b$ is restricted by $|R(b)|$.

In the case $|R(u)|=|R(ub)|$
the product $M_uE_b$ returns all images of nonzero columns of
$M_u$ to its origin place in $M_u$.
$M_uL_x=M_s \to M_uM_bM_{b^-}L_x=M_s$.
\end{lem}

\begin{proof}
For nonzero column $i$ of $M_b$ there exists a unit
in corresponding row $i$ of $M_b^-$.

The rows with free placing of units in $M_{b^-}$ correspond zero columns of $M_b$ and therefore could imply on the product
$M_bM_{b^-}=E_b$.
Besides, all $n$ units of columns $M_b$ are moved to $E_b$,
whence units of the arbitrary placing could not imply
on $E_b=M_bM_{b^-}$.

By Lemma \ref{l2}, the rank of $M_bM_{b^-}$ is not greater than the rank $|R(b)|$ of $M_b$.
Consequently rank of $E_b$ is restricted by $|R(b)|$.

In the case $|R(u)|=|R(ub)|$ the matrix $M_b$ does
not merge columns of $M_u$ and therefore
the product $M_uE_b$ returns all columns of $M_uE_b$
 to its origin place in $M_u$, whence
$M_uL_x=M_s \to M_uM_bM_{b^-}L_x=M_s$.

\end{proof}

\begin{lem} \label{l7}
For every equation $M_uL_x=M_s$ and every letter $\beta$

\begin{equation}
 M_{u\beta}L_y=M_s      \label{yx}
\end{equation}
for minimal solutions $L_y$ and $L_x$ with
$(S_q,y) \leq (S_q,x)$.

For  $|R(u)|\neq |R(u\beta)|$ and every letter $\beta$ there exists solution $L_y$ of the equation
$M_uM_{\beta}L_y=M_{u\beta}L_y=M_s$
such that $(S_q,y)<(S_q,x)$ for minimal solutions.

For some $L_y$, there exists a matrix $M_{\beta^-}$ such that $L_y=M_{\beta^-}L_x$.

$|R(u)|=|R(u\beta)|$ implies
$(S_q,y)=(S_q,x)$ for minimal solutions $L_y$ and $L_x$.
The invertible matrix $M_{\beta^-}$ does not change number
of units of $L_x$ in every column of
$L_y=M_{\beta^-}L_x$.
$M_uL_x=M_s \to M_uM_{\beta}M_{beta^-}L_x=M_s$.
  \end{lem}

\begin{proof}

The equality in (\ref{yx}) is correct for some $L_y$.
By Lemma \ref{l2} $|R(u)|\geq |R(u\beta)|$.
Therefore by Lemma \ref{l5} $(S_q,y) \leq (S_q,x)$ for minimal solutions $L_x$ and $L_y$.

From $|R(u)|\neq |R(u\beta)|$ due to Lemma \ref{l2} follows
$|R(u\beta)|<|R(u)|$, whence for some solution $L_y$
of the equation $M_uM_{\beta}L_y=M_s$
$(S_q,y)<(S_q,x)$ for both such minimal solutions by Lemma
\ref{l5}.

In the case $(S_q,y)<(S_q,x)$ for minimal solutions $L_x$ and
$L_y$ the column $i$ of the matrix $M_{\beta}$ merges
some columns $k_1, ..,k_m$ of $M_u$.
Then let us paste in pseudoinvers matrix $M_{\beta^-}$ for every   $i$ unit in the cell $k_1,i$ and units in the cells
$(k_r,j)$ ($r>1$) for distinct zero columns $j$ of $M_u$.
Therefore the matrix $E_{\beta}$ returns a part of nonzero columns
from $R(u)$ to the origin place in the matrix $M_uM_{\beta}$.
This part has all nonzero columns of $M_uM_{\beta}$, whence
$M_uL_x=M_s=M_uM_{\beta}M_{\beta^-}L_x$.

Let $|R(u)|=|R(u\beta)|$.

According to Definition \ref{inv}, the units in the matrix
$M_{\beta^-}$ in some rows can be disposed arbitrarily.
But they could not imply on the product
$M_{\beta}M_{\beta^-}=E_{\beta}$ by Lemma \ref{l9}.
$E_{\beta}$ returns nonzero columns
from $R(u)$ to the origin place also due to Lemma \ref{l9}.
Hence the equality in

\centerline{$M_uM_{\beta}M_{\beta^-}L_x=M_{u\beta}M_{\beta^-}L_x=
M_{u\beta}L_y=M_s$}
is correct in such case for every matrix $M_{\beta^-}$
and solution $L_y=M_{\beta^-}L_x$.

Moreover, the invertible matrix $M_{\beta^-}=M_{\beta^{-1}}$ keeps the number of units of every column of $L_x$ in
$L_y=M_{\beta^-}L_x$ by Remark \ref{r4},
whence also for column $q$ one has
$(S_q,y)=(S_q,x)$ for minimal solutions $L_y$ and $L_x$.

\end{proof}

\begin{cor}  \label{c9}
A set of linear independent solutions $L_y=M_{a^-}L_x$ of
(\ref{yx}) and $L_x$ with constant $(S_q,y)=(S_q,x)$
 (and therefore $R(y)=R(x)$)
can be created by help of invertible matrices $M_{a^-}$
by words $a$ of restricted length
(Lemma \ref{l5}, Remark \ref{ri}).

So such invertible matrices $M_{a^-}$ can create even a maximal subset of linear independent matrices $L_y$ with fixed
$(S_q,y)=(S_q,x)$ and greater in a space.

\end{cor}

\section{Examples}
The coordinate $j$ in $n$-vector of subset of states $c_u$ has unit if the state $j \in c_u$ and zero in opposite case.
For instance, $(011011)$ means that the subset has states $\{2,3,5,6\}$.

Units in $c_u$ correspond nonzero columns from $R(u)$ of matrix $M_u$.
The vector $c_u$ is equal to column $q$ of solution $L_x$ of equation $M_uL_x=M_s$.

The matrices $L_x$ corresponding word $u$ of $M_u$ (or $L_v$ where $L_x\sqsubseteq_q L_v$) of fixed $(S_g,x)$
in lines of examples below are linear independent.

J. Kari \cite{Ka} discovered the following example of $n$-state automaton
with minimal synchronizing word of length $(n-1)^2$ for $n=6$.

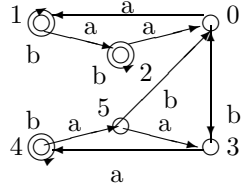
\begin{picture}(130,70)
 \end{picture}
\begin{picture}(130,74)
\multiput(6,60)(64,0){2}{\circle{6}}
\multiput(6,13)(64,0){2}{\circle{6}}
 \multiput(22,56)(22,0){2}{a}
\multiput(16,19)(34,0){2}{a}
 \put(36,21){\circle{6}}
\put(36,48){\circle{6}}
 \put(7,14){\vector(4,1){28}}
\put(7,57){\vector(4,-1){26}}

\put(39,52){\vector(4,1){27}}
 \put(37,20){\vector(4,-1){28}}
\put(67,63){\vector(-1,0){57}}
 \put(36,64){a}
\put(67,12){\vector(-1,0){57}}
 \put(32,0){a}

\put(70,15){\vector(0,1){42}}
 \put(70,59){\vector(0,-1){42}}
\put(34,21){\vector(1,1){36}}
 \put(52,28){b}

  \put(76,22){b}
\put(76,10){3}
\put(76,60){0}
\put(27,25){5}
\put(43,38){2}

\put(25,37){b}
\put(36,48){\circle{10}}
\put(-6,10){4}
\put(0,20){b}
\put(-6,60){1}
\put(0,45){b}
\put(37,42){\vector(2,1){4}}

\put(6,60){\circle{10}}
\put(4,64){\vector(2,1){4}}

 \put(6,13){\circle{10}}
\put(7,7){\vector(2,1){4}}
 \end{picture}

The minimal synchronizing word

\centerline{$s=\it ba^2bababa^2b^2aba^2ba^2baba^2b$}
has the length at the \v{C}erny border.

Every line below presents a pair (word $u$, $n$-vector $c_u$) of the word $u$.

$(b, 111110)$ $|R(u)|=5$

$(ba, 111011)$

$(ba^2, 111101)$

$(ba^2b, 111100)$ $|R(u)|=4$

$(ba^2ba, 111010)$

$(ba^2bab, 011110)$

$(ba^2baba, 101111)$ $|R(v)|=5$ (l01011 of $c_u\subset c_v$)

$(ba^2babab, 101110)$ $|R(u)|=4$

$(ba^2bababa, 110101)$

$(ba^2bababa^2, 011101)$

$(ba^2bababa^2b, 111000)$ $|R(u)|=3$

$(ba^2bababa^2b^2, 011100)$

$(ba^2bababa^2b^2a, 110111)$ $|R(v)|=5$ (101010 of $c_u\subset c_v$)

$(ba^2bababa^2b^2ab, 001110)$ $|R(u)|=3$

$(ba^2bababa^2b^2aba, 100011)$

$(ba^2bababa^2b^2aba^2, 011111)$  $|R(v)|=5$ (010101 of $c_u\subset c_v$)

$(ba^2bababa^2b^2aba^2b, 110000)$ $|R(u)|=2$

$(ba^2bababa^2b^2aba^2ba, 011000)$

$(ba^2bababa^2b^2aba^2ba^2, 101000)$

$(ba^2bababa^2b^2aba^2ba^2b, 001101)$ $|R(v)|=3$ (001100 of $c_u\subset c_v$)

$(ba^2bababa^2b^2aba^2ba^2ba, 100010)$ $|R(u)|=2$

$(ba^2bababa^2b^2aba^2ba^2bab, 000110)$

$(ba^2bababa^2b^2aba^2ba^2baba, 001011)$  $|R(v)|=3$ (000011 of $c_u\subset c_v$)

$(ba^2bababa^2b^2aba^2ba^2baba^2, 000101)$ $|R(u)|=2$

$(ba^2bababa^2b^2aba^2ba^2baba^2b=s, 100000)$  $|R(s)|=1$

By the bye, the matrices of right subwords of $s$ are simply linear independent.

It is possible that this property is by no means rare for minimal synchronizing word.
\\
\\
For the \v{C}erny sequence of $n$-state automata \cite{Ce}, \cite{La}, \cite{Li} the situation is more pure.

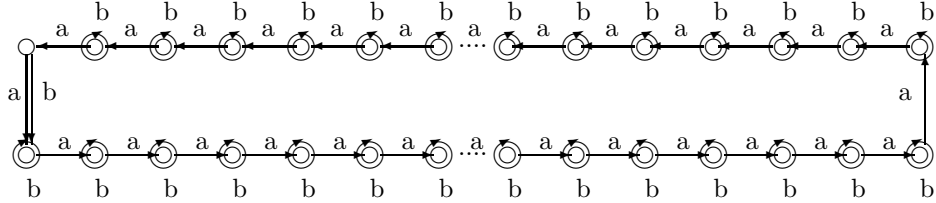
\begin{picture}(300,70)
\multiput(0,54)(26,0){14}{\circle{6}}
\multiput(26,54)(26,0){13}{\circle{10}}
\multiput(24,54)(26,0){6}{\vector(-1,0){20}}
\multiput(204,54)(26,0){6}{\vector(-1,0){20}}
\put(160,54){ ....}

\multiput(11,58)(26,0){13}{a}
 \multiput(26,64)(26,0){13}{b}
 \multiput(25,58)(26,0){13}{\vector(2,1){4}}
\put(340,17){\vector(0,1){34}}
 \put(0,51){\vector(0,-1){34}}
\put(2,51){\vector(0,-1){34}}
\put(-7,34){a} \put(6,34){b} \put(330,34){a}
\multiput(0,13)(26,0){14}{\circle{6}}
\multiput(0,13)(26,0){14}{\circle{10}}
\multiput(4,13)(26,0){6}{\vector(1,0){20}}
\multiput(186,13)(26,0){6}{\vector(1,0){20}}
\put(160,13){ ....}

\multiput(12,15)(26,0){13}{a}
 \multiput(0,-3)(26,0){14}{b}
 \multiput(-3,17)(26,0){14}{\vector(2,1){4}}
 \end{picture}
\\
\\
The minimal synchronizing word

\centerline{$s=b(a^{n-1}b)^{n-2}$}
 of the automaton also has the length at the \v{C}erny border.
\\
\\
For $n=4$

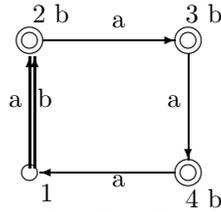
\begin{picture}(140,54)
\end{picture}
\begin{picture}(140,50)
\multiput(-21,60)(60,0){2}{\circle{6}}
\multiput(-21,60)(60,0){2}{\circle{10}}

 \multiput(-21,10)(60,0){2}{\circle{6}}

\put(39,10){\circle{10}}

 \put(-21,12){\vector(0,1){42}}
\put(-19,12){\vector(0,1){42}}

 \put(34,10){\vector(-1,0){51}}
\put(-16,60){\vector(1,0){50}}
 \put(39,55){\vector(0,-1){40}}

  \multiput(-29,35)(60,0){2}{a}
\multiput(10,5)(0,60){2}{a}

\put(-20,67){2 b}
 \put(-17,-1){1}
\put(38,67){3  b}
 \put(38,-3){4  b}
\put(-18,35){b}

 \end{picture}

and synchronizing word  $baaabaaab$ with pairs of word $u$ and
$n$-vector $c_u$ of linear independent matrices $L_u$ below.

$(b, 0111)$ $|R(u)|=3$

$(ba, 1011)$

$(baa, 1101)$

$(baaa, 1110)$

$(baaba, 1010)$  $|R(u)|=2$

$(baaaba, 0011)$

$(baaabaa, 1001)$

$(baaabaaa, 1100)$ $|u|=8$

$(baaabaaab=s, 0100)$ $|R(s)|=1$
\\
\\
In the example of Roman \cite{Ro}

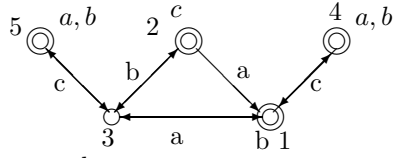
\begin{picture}(140,54)
\end{picture}
\begin{picture}(140,50)
\multiput(-21,39)(56,0){3}{\circle{6}}
\multiput(-21,39)(56,0){3}{\circle{10}}

 \multiput(6,10)(60,0){2}{\circle{6}}

\put(66,10){\circle{10}}
 \put(5,12){\vector(-1,1){24}}
  \put(-19,36){\vector(1,-1){24}}

\put(67,12){\vector(1,1){22}}
  \put(91,36){\vector(-1,-1){22}}

  \multiput(-16,20)(97,0){2}{c}

\put(7,12){\vector(1,1){24}}
\put(31,36){\vector(-1,-1){24}}

\put(38,36){\vector(1,-1){24}}

   \put(9,10){\vector(1,0){54}}
   \put(63,10){\vector(-1,0){54}}

\put(2,-1){3}
 \put(28,0){a}
\put(53,25){a}
 \put(11,24){b}

\put(-33,41){$5$}
\put(-14,45){$a,b$}
\put(28,48){$c$}
\put(19,41){$2$}
\put(98,45){$a,b$}
\put(88,47){$4$}
\put(60,-2){b   1}
 \end{picture}

the minimal synchronizing word

\centerline{$s=ab(ca)^2c$ $bca^2c$ $abca$}

The line below presents a pair (word $u$, $n$-vector $c_u$) of the word $u$.

$(a, 10111)$ $|R(u)|=4$

$(ab, 11011)$

$(abc, 11110)$

$(abca, 10110)$ $|R(u)|=3$

$(abcac, 10011)$

$(abcaca, 01111)$ $|R(v)|=4$  (00111 of $c_u\subset c_v$)

$(abcacac, 10101)$ $|R(u)|=3$

$(abcacacb, 11001)$

$(abcacacbc, 01110)$

$(abcacacbca, 10010)$ $|R(u)|=2$

$(abcacacbca^2, 00110)$

$(abcacacbca^2c, 10001)$

$(abcacacbca^2ca, 11101)$  $|R(v)|=4$ (00101 of $c_u\subset c_v$)

$(abcacacbca^2cab, 01001)$   $|R(u)|=2$

$(abcacacbca^2cabc, 01100)$

$(abcacacbca^2cabca=s, 10000)$ $|R(s)|=1$

\section*{Acknowledgments}
I would like to express my gratitude to Francois Gonze,
Dominique Perrin, Marie B{\'e}al, Akihiro Munemasa, Wit Forys, Benjamin Weiss, Mikhail Volkov and Mikhail Berlinkov
for fruitful and essential remarks throughout the study.

 \end{document}